\documentclass[12pt]{iopart}

\usepackage{cite}

\expandafter\let\csname equation*\endcsname=\relax
\expandafter\let\csname endequation*\endcsname=\relax
\usepackage{amsmath}

\usepackage{graphicx}
\usepackage{graphics}
\usepackage{iopams} 

\usepackage{color}

\newtheorem{theorem}{Theorem} 
\newtheorem{lemma}{Lemma}
\newtheorem{corollary}{Corollary}

\newenvironment{proof}[1][Proof]{\begin{trivlist}
\item[\hskip \labelsep {\bfseries #1}]}{\end{trivlist}}

\newcommand{\qed}{\nobreak \ifvmode \relax \else
      \ifdim\lastskip<1.5em \hskip-\lastskip
      \hskip1.5em plus0em minus0.5em \fi \nobreak
      \vrule height0.75em width0.5em depth0.25em\fi}

\newcommand{\sech}{\rm sech\,}

\def\beann{\begin{eqnarray*}}
\def\eeann{\end{eqnarray*}}

\def\be{\begin{equation}}
\def\eea{\end{eqnarray}}
\def\ee{\end{equation}}
\def\bea{\begin{eqnarray}}
\def\ea{\end{array}}
\def\ba{\begin{array}}

\newcommand{\bel}[1]{\begin{equation}\label{#1}}

\def\zzz{{\mathchoice {\hbox{$\sf\textstyle Z\kern-0.4em Z$}}
{\hbox{$\sf\scriptstyle Z\kern-0.3em Z$}}
{\hbox{$\sf\scriptscriptstyle Z\kern-0.2em Z$}}
{\hbox{$\sf\textstyle Z\kern-0.4em Z$}}}}

\usepackage{listings}
\usepackage{framed}
\usepackage{xcolor}
\colorlet{shadecolor}{gray!18}

\begin{document}

\title[The double hypergeometric series for the 2D anisotropic Ising model]
{The double hypergeometric series for the partition function of
  the  2D anisotropic Ising model}

\author{G. M. Viswanathan}

\address{Department of Physics
and {National Institute of Science and Technology of Complex Systems,}
Universidade Federal do Rio Grande do Norte, 59078-970
Natal--RN, Brazil}

\begin{abstract}
  In 1944 Lars Onsager published the exact partition function of the
  ferromagnetic Ising model on the infinite square lattice in terms of
  a definite integral. Only in the literature of the last decade,
  however, has it been recast in terms of special functions.  Until
  now all known formulas for the partition function in terms of
  special functions have been restricted to the important special case
  of the isotropic Ising model with symmetric couplings. Indeed, the
  anisotropic model is more challenging because there are two
  couplings and hence two reduced temperatures, one for each of the
  two axes of the square lattice.  Hence, standard special functions
  of one variable are inadequate to the task.  Here, we reformulate
  the partition function of the anisotropic Ising model in terms of
  the Kampé de Fériet function, which is a double hypergeometric
  function in two variables that is more general than the Appell
  hypergeometric functions.  Finally, we present hypergeometric
  formulas for the generating function of multipolygons of given
  length on the infinite square lattice, for isotropic as well as
  anisotropic edge weights.  For the isotropic case, the results allow
  easy calculation, to arbitrary order, of the celebrated series 
  found by Cyril Domb.
\end{abstract}

\maketitle

\section{Introduction}

Lars Onsager solved the ferromagnetic Ising model with nearest
neighbor interactions on the infinite square lattice in 1942 and
published the solution in 1944 in a breakthrough paper~\cite{onsager}.
Let $J_1$ and $J_1$ denote the couplings, $K_i= J_i/kT$ the reduced
temperatures, and $Z(K_1,K_2)$ the partition function per site of the
ferromagnetic Ising model on the square lattice.  Onsager's solution is
then given by
\be
\label{eq-eoirhoierjoierjgoierjoiwejoiwjeth}
\log (Z/2) =
\frac{1}{2 \pi^2}
\underset{~0}{\overset{~~\mbox{\small $\pi$}} {\iint}}
\log[\cosh 2K_1 \cosh 2K_2 -\sinh 2K_1 \cos \theta_1 - \sinh 2K_2 \cos \theta_2]
~d\theta_1d\theta_2~.
~  
\ee
Furthermore, let
\bea
2\kappa_1 &=& \tanh 2K_1 ~\sech 2K_2  \nonumber\\
2\kappa_2 &=& \tanh 2K_2  ~\sech 2K_1  \label{eq-kappa12}
~.
\eea
Then, Onsager showed that 
\smallskip 
\bea & & \nonumber 
\log Z -
\tfrac 1 2 \log (4 \cosh 2K_1 \cosh 2K_2) 
\\
& & 
= 
\frac 1 {2\pi^2}
\underset{0}
{\overset{~~\mbox{\small $\pi$}}
  {\iint}}
\log[1- 2\kappa_1 \cos \theta_1 - 2 \kappa_2 \cos \theta_2 ]
~d\theta_1d\theta_2
\label{eq-347t9823gh3f23}
 \\
& &
= - \frac 1 2  \sum_{j+k>0}^\infty
{
  (2j + 2k -1)!
\over 
(j!)^2 (k!)^2
}
\kappa_1^{2j}
\kappa_1^{2k} ~.
\label{eq-Z-aeiurhgiuerhguierg}
\eea
This last Eq.~(\ref{eq-Z-aeiurhgiuerhguierg}) is Eq.~(109b) in the
1944 paper~\cite{onsager}.  It is easily obtained  from
Eq.~(\ref{eq-347t9823gh3f23}) as follows: (i) first Taylor expand 
the logarithm in the Mercator series, (ii) then apply the binomial
theorem to expand powers of $(2\kappa_1 \cos \theta_1 + 2 \kappa_2
\cos \theta_2)$, and (iii) finally perform the double integral of the
now-separable integrands as an iterated integral, term by term. For
the case of the isotropic Ising model with symmetric couplings
$J_1=J_2=J$, we can write $K_1=K_2=K$ and $\kappa_1=\kappa_2=\kappa$
and the above series simplifies considerably (see Eq.~(109c) of
Ref.~\cite{onsager}).

In the literature of the last decade, the partition function for the
isotropic Ising model has been reformulated in terms of special
functions (see below). However, the anisotropic model is much more
challenging because there are two reduced temperatures --- one for
each coupling. Here we address this problem by more carefully studying
the series (\ref{eq-Z-aeiurhgiuerhguierg}).  We thereby obtain a
double hypergeometric reformulation of the partition function for the
2D anisotropic Ising model.

Hypergeometric functions belong to the family of special
functions~\cite{andrewsbook,sfagt}.  Much insight and understanding
can be gained when a definite integral or an infinite series is
re-expressed in terms of special functions, because the latter have
been widely studied and are often better understood than the former.
The analogy with elementary functions can be helpful to make this
point clear: \smallskip
\beann
\sum_{n=1}^\infty {(1-x)^n\over n} &= -\log (x)   ~, \quad\quad
& |x-1|<1   \\
\int_a^b {dx\over x} &= \log\left(b\over a\right)~, \quad\quad & {b\over a}>0 ~.
 \eeann
The left-hand sides of the above are precisely equal to the right-hand
sides.  However, the
appearance of the 
logarithm  on the
right gives us immediate insight and intuition that may not be so
obvious in the expressions to  the left.  The same kind of insight can
be gained from special functions.
For further details on the motivation for expressing
definite integrals and infinite series in terms of special functions,
see Refs. \cite {andrewsbook,jstat2015,cipra1998,sfagt,scirep,kdf2,hucht2011} .

In Section \ref{sec-hyper} we review the definition of the
prerequisite hypergeometric functions. In Section \ref{sec-res} we
state our main results, together with the proofs.
Section 
\ref{sec-concl} ends with discussion and conclusions.

\section{The Kampé de Fériet double hypergeometric function}
\label{sec-hyper}

The generalized hypergeometric function $_pF_q$ of a variable $x$ has
series expansion $\sum c_n x^n $ such that the ratio $c_{n+1}/c_n$ of
successive coefficients is a ratio of polynomials in $n$. In other
words, the ratio of coefficients is a a rational function of $n$. The
degrees of the polynomials of the numerator and denominator are $p$
and $q+1$ respectively.
 Let the
Pochhammer symbol $(x)_n$ denote the
rising factorial for 
$n=1,2,3\ldots$ ~as follows: 
\begin{eqnarray} 
(x)_0&=1~, \nonumber   \nonumber \\
  (x)_n&= (x)_{n-1} (x+n-1)  \nonumber   \\
  &= {\Gamma(x+n) \over \Gamma(x)}~.
\label{eq-pochwewegwgwrg}
\end{eqnarray}
Then 
$_pF_q$ is defined as follows~\cite{andrewsbook}:
    \bel{eq-pfq-def}
_pF_q\left[ \ba{c}{a_1,a_2,\dots, a_p} \\ {b_1,b_2,\dots,b_q}\ea ;x \right]
= \sum_{n=0}^\infty
{(a_1)_n(a_2)_n \dots (a_p)_n  \over (b_1)_n(b_2)_n \dots (b_q)_n  }
 ~{x^n \over
         {n!}}
~.  
\ee 
For $p<q+1$ the function is entire.  In contrast, for $p>q+1$ it has
zero radius of convergence. For $p=q+1$ the radius of convergence is
1.

The isotropic Ising model can be expressed in terms of $_pF_q$
functions (see below).  In contrast, the expression
(\ref{eq-Z-aeiurhgiuerhguierg}) for the anisotropic Ising model contains a power
series in two different variables, hence we will need a 
more general and more powerful 
hypergeometric function.

The Kampé de Fériet function
$F^{p,q}_{r,s}$
is a double hypergeometric function of
higher order in 2 variables~\cite{kdf1,kdf2} that 
generalizes the $_p F_q$ function.
It is defined as  
\smallskip
\bea 
& & 
F^{p,q}_{r,s}
\left [~
\renewcommand\arraystretch{1.4} 
  \ba{c|c|}
  a_1\ldots a_p & ~b_1, b'_1 ~; ~\ldots~; ~ b_q, b'_q \\
    c_1\ldots c_r & ~  d_1, d'_1 ~; ~\ldots~; ~ d_s, d'_s  
\ea~~  x,y
\right] \nonumber 
\\
~
& & ~~
= ~
\sum_{m,n=0}^\infty
{
 ~\! \prod_{k=1}^p (a_k)_{m+n} ~~  \prod_{k=1}^q \big [ (b_k)_{m} (b'_k)_{n}  \big]
 \over
 ~\prod_{k=1}^r (c_k)_{m+n}  ~~ \prod_{k=1}^s  \big [  (d_k)_{m} (d'_k)_{n} \big]
} ~{x^m y^n \over m!n!}
~~.
\label{eq-kdf-def}
\eea

\noindent  
According to convention, if any of the quantities $p$ or $q$ in
(\ref{eq-pfq-def}) or any of $p,~q,~r,~s$  in (\ref{eq-kdf-def}) are
zero, then a   dash is used to indicate the absence of the
corresponding parameters, e.g.,
\[
\exp(x) = {_0}F_0
\left[ \ba{l} - \\ - \ea; x\right]~.
\]

We note in passing that for specific low values of the parameters
$p,~q,~r,~s$, the Kampé de Fériet function can be reduced to the 
simpler Appell hypergeometric functions,  of which there are four. For
the particular values of the parameters $p$ and $q$ for the anisotropic
Ising model (see below), this  standard reduction to the Appell
functions is not possible, so we do not further discuss this point
here.

\section{Results and proofs}
\label{sec-res}

Our first  result is the natural continuation of previous results.
In the 1970s,  Glasser and Onsager obtained the following
expression for the partition function of the 2D isotropic Ising
model~\cite{jstat2015,pc}: \be
  \ln Z(K) = \ln(2 \cosh 2K) 
\label{eq-ons-gla}
- {1\over 2} 
+ {1\over\pi} {\mathsf E(4 \kappa)} 
+ \kappa^2 ~{_4}F_3
\left[ \ba{c}{{1\over 2},1,1,{3\over 2}} \\ {2,2,2}\ea 
; 16 \kappa^2 \right]  
~.  
\ee
where ${\mathsf E(\cdot)}$ denotes the complete elliptic integral of
the second kind that can also be written as a $_2F_1$ function
\cite{jstat2015,andrewsbook}.
     { In 2011, a simpler hypergeometric formula was
       discovered by Hucht, Grüneberg and Schmidt using Wolfram
       Mathematica~\cite{hucht2011}. In 2015 the identical expression
       was independently found and a derivation was
       given~\cite{jstat2015}:}
\be
\label{eq-0gu32905u0w945u90u$-4F3}
\ln  Z 
= \ln(2 \cosh 2K) 
- \kappa^2 ~{_4}F_3
\left[ \ba{c}{1,1,{3\over 2},{3\over 2}} \\ {2,2,2}\ea 
; 16 \kappa^2 \right]  ~.
\ee
The above two expressions are related via a hypergeometric identity
(see Eq. (33) in Ref.  \cite{jstat2015}).  In 2016, the exact low
temperature series was expressed as an 
infinite series in
terms of complete and partial Bell polynomials~\cite{scirep}.
All
the above  results are restricted to  the isotropic case and do not apply
to the anisotropic Ising model.  Here we address this limitation and
report an advance for the anisotropic case:

\begin{theorem}
  \label{th1}
  Let $\kappa_1$ and $\kappa_2$ be given by (\ref{eq-kappa12}). Then
  the partition function $Z(K_1,K_2)$ of the anisotropic Ising model
  as defined in (\ref{eq-eoirhoierjoierjgoierjoiwejoiwjeth}) is given by
  \bea & &  \nonumber  
\log Z  - \tfrac 1 2
\log (4 \cosh 2K_1 \cosh 2K_2)  
\\ \nonumber  
& & 
= 
-
3 ~  \kappa_1^2 \kappa_2^2 ~
{
F^{2,1}_{0,2}
\left [~
  \ba{c|c|}
  2,5/2 & ~1,1~  \\
    \mbox{--} & ~  2,2; 2,2  
\ea~~  4\kappa^2_1, 4\kappa^2_2~
\right]
}  
\\
& & 
~~~ - \frac 1 2
~  \kappa_1^2  ~
_3F_2  
\left[
  \ba{c}{1,1,3/2} \\ 2,2\ea ;4\kappa_1^2 ~\right]
~  - \frac 1 2 \kappa_2^2  ~
_3F_2 
\left[
  \ba{c}{1,1,3/2} \\ 2,2\ea ;4\kappa_2^2 ~\right]~. 
\label{qewegweg203841029384}
\eea
\end{theorem}

\begin{corollary}
\label{cor-wgowirjgoeirgewrg}
  \bea & & 
\log Z - \tfrac 1 2
\log (4 \cosh 2K_1 \cosh 2K_2)  
\nonumber  
\\
& & \nonumber
= 
-
3 ~  \kappa_1^2 \kappa_2^2 ~
{
F^{2,1}_{0,2}
\left [~
  \ba{c|c|}
  2,5/2 & ~1,1~  \\
    \mbox{--} & ~  2,2; 2,2  
\ea~~  4\kappa^2_1, 4\kappa^2_2~
\right]
}
\\
& & 
~~~ + \frac 1 2
\left[\,
  \log \left(\frac{1}{2} \left(1+\sqrt{1-4 \kappa_1^2}\right)\right)
+  \log \left(\frac{1}{2} \left(1+\sqrt{1-4 \kappa_2^2}\right)\right)
\,  \right]~.
\label{eqeq-34753987698ng}
\eea
\end{corollary}

\medskip
\medskip
\medskip

\medskip
\medskip
\medskip

To prove the above, we will use the following supporting lemmas:

\begin{lemma}[Basic identities for rising factorials]
  Let $(a)_n$  be defined according to (\ref{eq-pochwewegwgwrg}). Then,
   \label{lemma1}
\bea
n!&=~&(1)_n  
~,
\\
(n+1)! &=~& (2)_n ~,\\
a+n~ &=~&  a {(a+1)_n \over  (a)_n }~,  \\
(2n+1)!~ &=~& 2^{2n} (3/2)_n (1)_n \label{eq-12414141}  
~,
\\
(2j+2k+3)!~ &=~&
{6 \cdot 2^{2(j+k)} 
(2)_{j+k}  (5/2)_{j+k} } ~.
\label{eq23r2t}  
\eea

\end{lemma}
'

\begin{proof}
  The first three above follow immediately from the definition of the
  rising factorial.  The claim (\ref{eq-12414141}) can be proven by
  splitting the terms in the factorial into odd and even factors and
  dividing each factor by 2:
\bea
(2n+1)! & =& [1] \times 
(3)(5)(7)  \ldots (2n+1) \times (2)(4)(6)\ldots (2n)
\nonumber
\\
&=&
2^{2n}~(3/2)_n (1)_{n}  ~.  
\eea
The claim (\ref{eq23r2t}) is slightly more complicated but follows
from the previous 4 claims:
\bea
(2j+2k+3)!
& & =  (2j+2k+1)! ~(2 j + 2 k+2) (2j + 2k+3)
\\
& & \nonumber
= 2^{2j+2k}(3/2)_{j+k}(1)_{j+k} ~
(2 j + 2 k+2)(2 j + 2 k+3)
\\ \nonumber
& & 
=
2^{2j+2k+2}(3/2)_{j+k}(1)_{j+k}
(j + k+1)(j + k+3/2)
\\
 & & 
= {2^{2j+2k+2}(3/2)_{j+k}(1)_{j+k} 
 (2)_{j+k} (3/2) (5/2)_{j+k} \over (1)_{j+k}(3/2)_{j+k}}
\\
 & & 
=
{6 \cdot 2^{2(j+k)} 
(2)_{j+k}  (5/2)_{j+k} } ~,   
\eea
which completes the proof.
\hfill $\square$
\end{proof}

Our next lemma deals with the infinite double series in
(\ref{eq-Z-aeiurhgiuerhguierg}).
For convenience let us 
write 
the partition function in 
(\ref{eq-Z-aeiurhgiuerhguierg})
as 
\be 
\label{q234235weqweg}
\log Z -
\tfrac 1 2 \log (4 \cosh 2K_1 \cosh 2K_2)  
= - \frac 1 2 \Phi(\kappa_1,\kappa_2)~,
\ee
with 
$\Phi$ thus  defined as 
\be
\Phi(x_1,x_2)
= \sum_{j+k>0}^\infty
{
  (2j + 2k -1)!
\over 
(j!)^2 (k!)^2
}
x_1^{2j}
x_1^{2k} ~.
\label{eq-eorijgeoirjgoierjgioerjgijeroigj}
\ee

\begin{lemma}
  \label{lemma2}  
  \bea 
\Phi(x_1,x_2) \nonumber  
& =  & 
6 ~  x_1^2 x_2^2 ~
{
F^{2,1}_{0,2}
\left [~
  \ba{c|c|}
  2,5/2 & ~1,1~  \\
    \mbox{--} & ~  2,2; 2,2  
\ea~~  4x^2_1, 4x^2_2~
\right]
}
\\
& & 
+  x_1^2  ~
_3F_2 
\left[
  \ba{c}{1,1,3/2} \\ 2,2\ea ;4x_1^2 ~\right]
+  x_2^2  ~
_3F_2 
\left[
  \ba{c}{1,1,3/2} \\ 2,2\ea ;4x_2^2 ~\right] ~. 
\nonumber \\ 
\eea
  
  \end{lemma}

\begin{proof}
If the sum in (\ref{eq-eorijgeoirjgoierjgioerjgijeroigj}) started from
$j=k=0$, then  the task would be easier, but the sum excludes the
point $j=k=0$. So we will rewrite $\Phi$ as a
sum of a double series for $j>0,~k>0$ 
and separate 
single series corresponding to $j=0$ and
$k=0$:
\bea 
\Phi(x_1,x_2)
& =    & \sum_{j,k=1}^\infty
{
  (2j + 2k -1)!
\over (j!)^2 (k!)^2
}
x_1^{2j}
x_1^{2k}
  +
\sum_{j=1}^\infty
{
  (2j -1)!
\over 
(j!)^2 
}
(x_1^{2j}
+
x_2^{2j})
\eea
We next change the summation indices to start at zero and invoke the identities in lemma
\ref{lemma1}:
\bea \nonumber
\Phi(x_1,x_2)
& =    & \sum_{j,k=0}^\infty
{
  (2(j+k)) +3)!
\over 
((j+1)!)^2 ((k+1)!)^2
}
x_1^{2(j+1)}
x_1^{2(k+1)} 
\\ 
& & 
+
\sum_{j=0}^\infty
{
  (2j +1)!
\over 
((j+1)!)^2 
}
(x_1^{2(j+1)}
+
x_2^{2(j+1)}) \nonumber 
\\ \nonumber
& =    &
x_1^2 x_2^2
\sum_{j,k=0}^\infty
{
6 \cdot 2^{2(j+k)} ~ (2)_{j+k}~(5/2)_{j+k}
\over 
((2)_j)^2 ~((2)_k)^2
}
x_1^{2j}
x_1^{2k}
\\  \nonumber
& & 
+
x_1^2 
\sum_{j=0}^\infty
{
2^{2j}(3/2)_j (1)_j
\over 
((2)_j)^2 
}
x_1^{2j}
\\
& & 
+
x_2^2
\sum_{j=0}^\infty
{
2^{2j}(3/2)_j (1)_j
\over 
((2)_j)^2 
}
x_2^{2j}
~.
  \label{eq-nuieniuwhgi4838438}
\eea
Finally, from the definitions (\ref{eq-pfq-def}) of the generalized hypergeometric
function and (\ref{eq-kdf-def}) of the Kampé de Fériet function,
(\ref{eq-nuieniuwhgi4838438})
becomes 
\bea
\nonumber
{\Phi(x_1,x_2) }
& =    &
6 ~  x_1^2 x_2^2 ~
{
F^{2,1}_{0,2}
\left [~
  \ba{c|c|}
  2,5/2 & ~1,1~  \\
    \mbox{--} & ~  2,2; 2,2  
\ea~~  4x^2_1, 4x^2_2~
\right]
} 
\\
& & 
+~  x_1^2  ~
_3F_2 
\left[ \ba{c}{1,1,3/2} \\ 2,2\ea ;4x_1^2 ~\right]
+~  x_2^2  ~
_3F_2 
\left[ \ba{c}{1,1,3/2} \\ 2,2\ea ;4x_x^2 ~\right] ~,~
\label{eq-Phi}
\nonumber
\\
\eea
which completes the proof.  \hfill $\square$
\end{proof}

\begin{proof}[Proof of theorem \ref{th1}]

The claim follows from using lemma \ref{lemma2} in  
(\ref {q234235weqweg}).
\hfill $\square$
\end{proof}

\begin{proof}[Proof of corollary \ref{cor-wgowirjgoeirgewrg}]
  To prove the claim (\ref{eqeq-34753987698ng}) starting from theorem \ref{th1}, it suffices to show that 
\be
\label{eq-09ire0909ihgf}
x~_3F_2 
\left[
  \ba{c}{1,1,3/2} \\ 2,2\ea ;x ~\right]
=
- {4 \log \left(\frac{1}{2} \left(\sqrt{1-x}+1\right)\right)} ~.
\ee
We begin with the Mercator series and  Newton's generalized binomial theorem:
\bea
\log (1+\sqrt{1-x})
&
=&
- \sum_{n=1}^\infty 
{ (-1)^n
(1-x)^{n/2}
\over n}
\\
&=&
- \sum_{n=1}^\infty  
{ (-1)^n 
  \over n}
\sum_{j=0}^\infty
{n/2 \choose j} (-x)^j
\\
&=&
-  \sum_{j=0}^\infty (-x)^j \sum_{n=1}^\infty  
{ (-1)^n 
  \over n}
{n/2 \choose j} ~.
\eea
Next, 
we rewrite  the binomial coefficient more conveniently as 
\be
{n/2 \choose j } =  { (n/2 +1-j)_j  \over j!} ~.
\ee
We then get  
\bea
 \log (1+\sqrt{1-x})
 \\
=
-  \sum_{j=0}^\infty (-x)^j
\sum_{n=1}^\infty  
{ (-1)^n 
  \over n}
{ (n/2 +1-j)_j  \over j!}
\\
=
-  \sum_{j=0}^\infty
{ (-x)^j\over j!}
\sum_{n=1}^\infty  
{ (-1)^n 
  \over n}
{ (n/2 +1-j)_j }
\\
=
-  \sum_{j=0}^\infty
{(-x)^j
\over j!}
\sum_{n=1}^\infty  
\left[
{ (n +1-j)_j
  \over 2n}
-
{ (n +1/2-j)_j
  \over 2n-1}
\right]~. \label{eq-j-298r72987r92837r}
 ~~~~~~ ~~~
\eea
The sum over $n$ can be done after transforming the rising factorials
to gamma functions:
\bea & & 
\sum_{n=1}^\infty  
\left[
{ (n +1-j)_j
  \over 2n}
-
{ (n +1/2-j)_j
  \over 2n-1}
\right]
\\
& & 
=
\left[
\frac{\, \Gamma (-j)}{2 \, \Gamma (1-j)^2}-\frac{\sqrt{\pi } \, \Gamma (-j)}{2 \, \Gamma \left(\frac{1}{2}-j\right) \, \Gamma (1-j)}
\right]
\\
& &
=
\left[
\frac{\sqrt{\pi } \cos (\pi  j) \, \Gamma \left(j+\frac{1}{2}\right)-\sin (\pi  j) \, \Gamma (j)}{2 \pi  j}
\right]~.
\label{eq-deno-wirgjrg}
\eea
where in the last step we have used Euler's reflection formula.
The terms for $j>0$ in (\ref{eq-j-298r72987r92837r}) can be evaluated
as usual, but the term for $j=0$ is indeterminate due to the $j$ in
the denominator in (\ref{eq-deno-wirgjrg}), so it needs to be
evaluated as a limit. We take the limit $z\to j=0$ as follows:
\be
\lim _{z \to 0}
\frac{\sqrt{\pi } \cos (\pi  z) \, \Gamma \left(z+\frac{1}{2}\right)-\sin (\pi  z) \, \Gamma (z)}{2 \pi  z} =  - \log 2~.
\ee
For $j\neq 0$ the term with $\sin (\pi j)$ vanishes. We can then simplify the expression
using  Legendre's duplication formula for  the gamma function
and applying lemma \ref{lemma1}:

\bea
\log (1+\sqrt{1-x})
&=&
\log 2 - 
\sum_{j=1}^\infty
{1\over j!}
\frac{\sqrt{\pi } \cos (\pi  j) \, \Gamma \left(j+\frac{1}{2}\right)}{2 \pi  j}
(-x)^j
\\
\log (\tfrac 1 2 (1+\sqrt{1-x}))&=&
-
\sum_{j=1}^\infty
{1\over j!}
\frac{\sqrt{\pi } \, \Gamma \left(j+\frac{1}{2}\right)}{2 \pi  j}
x^j
\\
&=&
-
\sum_{j=1}^\infty
\frac{ 2^{1-2j}  
\, \Gamma( 2j) }
{2   j \, \Gamma(j) j!}
x^j
\\
&=&
- x 
\sum_{j=0}^\infty
\frac{ 2^{-2j-2}  
\, \Gamma( 2j+2) }
{   (j+1)!^2}
x^j
\\
&=&
- \frac x 4  
\sum_{j=0}^\infty
\frac{ 2^{-2j}  
(2j+1)! }
{   (j+1)!^2}
x^j
\\
&=&
- \frac x 4  
\sum_{j=0}^\infty
\frac{ (3/2)_n (1)_n 
}
{   (2)_n^2}
x^j
\\ & & = - \frac x 4  ~
{_3F_2} 
\left[
  \ba{c}{1,1,3/2} \\ 2,2\ea ;x ~\right] ~,
\eea

\noindent 
thereby establishing (\ref{eq-09ire0909ihgf}),  which 
completes the proof. 
\hfill $\square$
\end{proof}

Theorem \ref{th1} allows the generating function of multipolygons on
the infinite square lattice to be reformulated in terms of
hypergeometric functions.  By multipolygon on the infinite square
lattice is meant a connected or disconnected simple graph all of whose nodes
have even degrees.  Let $t_1$ be the weight of edges in the first
lattice axis direction and $t_2$ be the edge weight for the other
direction. For convenience, we will say that $t_1$ is the weight of
edges of the first type (say, horizontal edges) and $t_2$ the weight
of edges of the second type (say, vertical edges).

Let $\Lambda_N(t_1,t_2)$ be the generating function for
the the number of multipolygons on a finite square lattice with
$N=M^2$ sites, such that the coefficient of the term of degree $t_1^m
t_2^n$ gives the number of multipolygons with $m$ edges of the
first type and $n$ edges of the second type.
The case $t_1=t_2$, which we will call the isotropic case, has been
widely studied~\cite{feynman}. The anisotropic case $t_1\neq t_2$ has
received somewhat  less attention~\cite{montroll}.

We define the generating function $\Lambda(t_1,t_2)$ for
anisotropic multipolygons on  the infinite square lattice 
in terms of a  limit when the number of sites goes to infinity:
\be
\Lambda(t_1,t_2) = \lim_{N \to \infty} [ \Lambda_N(t_1,t_2) ] ^{1/N} \label{eq-rgehruehrg}
~.
\ee

The connection between the function \mbox{$\Lambda(t_1,t_2) = 1 +
  t_1^2 t_2^2 + t_1^2 t_2^4 + t_1^4 t_2^2 \ldots $~} and the
low-temperature series for the partition function of the 2D Ising
model is well known.  Moreover, the self-dual property of the Ising
model on the square lattice leads to a similar relation for the
high-temperature series.  Let
\bea u_i&=&\tanh(K_i)~,\\ v_i&=&\exp (-2K_i)~,\eea 
for $i=1,2$ be high- and low-temperature variables respectively.  Then
the following are 
well known~\cite{feynman,montroll}:
\bea
\label{eq-wrheurg}
Z &=& 2  \cosh K_1  \cosh K_2 ~  \Lambda(u_1,u_2)     \\
 \label{eq-wrheurg2}
Z &=&   \exp (K_1)\exp (K_2) ~ \Lambda(v_1,v_2)  ~.
\eea
From either of these two expressions, we
can obtain the double hypergeometric formulation of 
$\Lambda(t_1,t_2)$:

\begin{theorem}
  Let $\Lambda(t_1,t_2)$ be the
  generating function for multipolygons with
anisotropic 
  weights $t_1$ and $t_2$
    for the two directions on the infinite square lattice, defined according to
    (\ref{eq-rgehruehrg}).
  Let $\Phi$ be given by (\ref{eq-Phi}).
    Then, 
  \bea
\log \Lambda(t_1,t_2) &=& ~\frac 1 
2  \log \left[ {\left( {t_1}^2+1\right) \left( {t_2}^2+1\right)}\right] 
\nonumber
\\
& & 
- \frac 1 2 
\Phi
\left[
\frac{  {t_1} \left(1- {t_2}^2\right)}{\left( {t_1}^2+1\right) \left( {t_2}^2+1\right)},
\frac{  {t_2} \left(1- {t_1}^2\right)}{\left( {t_1}^2+1\right) \left( {t_2}^2+1\right)}
\right]  
~. 
\label{eq-lambda1231123}
\eea

\label{th-lambda}
\end{theorem}

\begin{proof}
  We can arrive at the claim through either the low- or the
  high-temperature formulations, due to the self-dual property. For
  completeness, we show both approaches.
    First note that
  \be
{  \sqrt{4 \cosh 2K_1 \cosh 2K_2}
  \over 2 \cosh K1 \cosh K1}
    =
     \sqrt{\left( {u_1}^2+1\right) \left( {u_2}^2+1\right)}~.  
\label{eq-wrergj998u}
     \ee
Similarly,
\be
\label{eq-wrergj998v}
{  \sqrt{4 \cosh 2K_1 \cosh 2K_2}
  \over  \exp (2K_1)  \exp (2K_1)}
    =
     \sqrt{\left( {v_1}^2+1\right) \left( {v_2}^2+1\right)}~.
     \ee
  Next observe that $\kappa_1$ and $\kappa_2$ can be expressed in terms of
  the high- and low-temperature variables as 

  \bea
  \kappa_1 = 
  \frac{  {u_1} \left(1- {u_2}^2\right)}{\left( {u_1}^2+1\right) \left( {u_2}^2+1\right)}
  =
\frac{ {v_2} \left( {v_1}^2-1\right)  }{\left( {v_1}^2+1\right) \left( {v_2}^2+1\right)}~,
\\
  \kappa_2 = 
  \frac{  {u_2} \left(1- {u_1}^2\right)}{\left( {u_2}^2+1\right) \left( {u_1}^2+1\right)}
  =
  \frac{ {v_1} \left( {v_2}^2-1\right)  }{\left( {v_2}^2+1\right) \left( {v_1}^2+1\right)}~.
  \eea
  The claim
(\ref{eq-lambda1231123})
  follows from substituting these expressions for $\kappa_i$ into
(\ref{q234235weqweg})  and then using 
    either 
    (\ref{eq-wrheurg}) and 
(\ref{eq-wrergj998u}) for the high-temperature variable 
    or else   (\ref{eq-wrheurg2}) with 
(\ref{eq-wrergj998v})
    for the low-temperature variable.  \hfill $\square$
\end{proof}

The above result simplifies considerably for the isotropic case
$t_1=t_2=t$.
The next result 
is known to experts in the
field,  but for some reason never seems to have been published
as such:

\begin{corollary}
  \label{coro01}
  Let $\Lambda(t)$  denote the generating function
$\Lambda(t,t)$
  of multipolygons on
  the infinite square lattice with isotropic weights $t$. Then, 
  \be
  \log \Lambda(t) = \log  \left( {t}^2+1\right) 
  - 
\frac{  t^2 \left(t^2-1\right)^2}{\left(t^2+1\right)^4}
  ~{_4}F_3
\left[ \ba{c}{1,1,{3\over 2},{3\over 2}} \\ {2,2,2}\ea 
  ; 
\frac{16 t^2 \left(t^2-1\right)^2}{\left(t^2+1\right)^4}
  \right]  
  \label{eq-wroigjeroigjoierjg298759837o}
  ~.
  \ee

\end{corollary}

\begin{proof}
  The claim follows from 
(\ref{eq-0gu32905u0w945u90u$-4F3}),  
(\ref{q234235weqweg}) and
theorem \ref{th-lambda} with
$t_1=t_2=t$.
 \hfill $\square$
\end{proof}

\bigskip

\section{Discussion and conclusion}
\label{sec-concl}

We first comment on a possible avenue for further research.  There is
a deep connection between the isotropic Ising model on the one hand
and spanning trees on the infinite square lattice on the other hand.
Guttmann and Rogers have defined a generating function that
generalizes the spanning tree constant~\cite{gr}.  Subsequently, it
was shown that this spanning tree generating function is related to
the partition function of the isotropic Ising model in a precise
way~\cite{pre2017}. The connection is due to the fact that the same
Mahler measure appears in both, which in turn is due to the random
walk structure function for the square lattice.  A natural question
now arises. Let us assume that one allows the spanning tree generating
function to have anisotropic weights in some suitably defined
manner. Then, does the known connection between spanning trees and the
isotropic Ising model generalize to the anisotropic model?  We
tentatively believe that the answer is yes,  and hope that others will
take interest in finding a conclusive answer.

We also comment on how the results were originally
intuited prior to the discovery of the actual proofs.  The results
above were not first proven or even intuited using  computer algebra
systems, theorem-proving software, artificial intelligence, machine
learning, and so forth.  Wolfram Mathematica software, for example,
cannot evaluate the integral in (\ref{eq-347t9823gh3f23}) and 
is unable to recognize the series for the definition
$\Phi$ in (\ref{eq-eorijgeoirjgoierjgioerjgijeroigj}).
In fact, the Kampé de Fériet function (\ref{eq-kdf-def}) is not part
of the repertoire of special functions included in Mathematica (at the
time of this writing).  Instead, theorem~\ref{th1} was the result of
traditional mathematical detective work, starting from (\ref
{eq-Z-aeiurhgiuerhguierg}).  After re-expressing the infinite double
series in the form (\ref{eq-eorijgeoirjgoierjgioerjgijeroigj}) of
$\Phi$, we calculated the ratio of successive terms of the series, in
the two variables separately and also together.  We thus saw that
these ratios are rational functions of the degrees of the two
variables --- precisely the mathematical signature of hypergeometric
series.  It was then just a matter of finding the suitable
hypergeometric function. We first looked at the Appell functions but
they are not of high enough order.  The Kampé de Fériet function was
the next natural candidate. What made the task slightly more difficult
is that both the Appell and Kampé de Fériet functions 
are double hypergeometric and therefore not
found in the textbooks.  For example, the {\it Special Functions} by
Andrews, Askey and Roy~\cite{andrewsbook} and {\it Special Functions:
  A Graduate Text} by Beals and Wong~\cite{sfagt} are both considered
to be authoritative texts, but neither book deals with the Appell or
Kampé de Fériet functions. Once the relevant definitions were found,
however, the proofs came quite naturally.

Finally, we  note that
theorem \ref{th-lambda} and corollary \ref{coro01}
allow 
very easy explicit evaluation of 
the multipolygon generating function to arbitrary order.
For the anisotropic case we get from (\ref{eq-lambda1231123})
\bea 
\Lambda(t_1,t_2) &=&
\nonumber 
1 ~+~ t_1^2 t_2^2 ~+~ t_1^4 t_2^2 ~+~ t_1^2 t_2^4   
\\
& & \nonumber 
~+~
t_1^2 t_2^6
 ~+~ 3 t_1^4 t_2^4
~+~t_1^6 t_2^2  
\\ 
& & \nonumber 
~+~ t_1^2 t_2^8
~+~ 6 t_1^4 t_2^6
 ~+~ 6 t_1^6 t_2^4
~+~ t_1^8 t_2^2   
\\ & &
~+~ t_1^2 t_2^{10}
~+~ 10 t_1^4 t_2^8
~+~ 22 t_1^6 t_2^6
~+~ 10 t_1^8 t_2^4 
~+~ t_1^{10} t_2^2 + \ldots
\eea
For the isotropic case, the formula
(\ref{eq-wroigjeroigjoierjg298759837o}) generates, to arbitrary order,
the celebrated series for the generating function of multipolygons
found by Cyril Domb~\cite{domb} in 1949 and whose coefficients
have  been incorporated into the On-Line Encyclopedia of Integer
Sequences (OEIS)~\cite{enc}:
\bea
\Lambda(t) &=&
1+t^4+2 t^6+5 t^8+14 t^{10}+44 t^{12}+152 t^{14}+566 t^{16}
\nonumber 
\\
& &
\nonumber 
+2234 t^{18}+9228 t^{20}+39520 t^{22}+174271 t^{24}+787246 t^{26}
\\
& & 
+3628992 t^{28}+17019374 t^{30}+O(t^{32})
~.
\eea
In fact, 
Eq. (\ref{eq-wroigjeroigjoierjg298759837o}) leads to very efficient
computation of the series. 
At the time of this writing, OEIS lists the following Wolfram
Mathematica code for generating 25 terms of the sequence A002890
\cite{enc}:
\begin{shaded}
\begin{verbatim}
(*For 25 terms,a PC computation lasts less than half an hour*)
m = 48 (*max y exponent*);
coes = CoefficientList[Series[
    Log[(1 + y^2)^2 - 2*y*(1 - y^2)*
       (Cos[2*Pi*u] + Cos[2*Pi*v])], {y, 0, m}], y] // Rest; 
nint[f_, {n_}] := 
 If[n == 2 || OddQ[n], 0, Print[n] ; 
  Integrate[Integrate[f, {u, 0, 1}], {v, 0, 1}]]; 
fy = MapIndexed[nint, coes].Table[y^k, {k, 1, m}]; 
CoefficientList[Series[Exp[fy/2], {y, 0, m}], y^2] 
(*Jean-François Alcover, Mar 19 2013*) 
\end{verbatim}

\end{shaded}
\noindent
The above calculation takes several minutes on a PC computer
\cite{enc}, as mentioned in the commented text on the first line of
code.  The reason that it takes so much time to calculate 25 terms is
that the above code depends on explicit integration of Onsager's
formula to obtain the sequence.  In contrast, Eq.
(\ref{eq-wroigjeroigjoierjg298759837o}) is ``already integrated,''
hence it should allow very much faster computation of the same
sequence.  In the following code below, we have used
(\ref{eq-wroigjeroigjoierjg298759837o}) and the Mathematica
implementation of the $_p F_q$ hypergeometric function to generate the
identical coefficients to the code above:

\begin{shaded}
\begin{verbatim}
CoefficientList[Series [E^(-((   t^2 (-1 + t^2)^2
    HypergeometricPFQ[{1, 1, 3/2, 3/2}, {2, 2, 2}, 
     (16 t^2 (-1 + t^2)^2)/(1 + t^2)^4])/(1 + t^2)^4)) (1 + t^2), 
     {t,0, 48}] , t^2] 
(* GM Viswanathan 2021 *)
\end{verbatim}

\end{shaded}
\noindent The latter code takes less than 0.1 seconds to generate 25
terms of the series, whereas the former code takes several minutes.
The speedup is by a factor larger than $10^3$.


\ack

We thank Wolfram
Research and Daniel Pryjma for a courtesy license for Mathematica.
We thank
 the anonymous referees
for helpful suggestions, and T.~C.~Adbias and 
H.~D.~Jennings for feedback.
This work was supported 
by CNPq (grant no. \mbox{302051/2018-0}).


\section*{References}



\begin{thebibliography}{mt1}


\bibitem{onsager}
Onsager L 1944
{Crystal Statistics. I. A Two-Dimensional Model with an Order-Disorder Transition},
{\it Phys. Rev. }
{\bf 65} 117 









\bibitem{andrewsbook} 
Andrews G E, 
  Askey R, 
and  Roy R  1999
{\it
  Special functions} 
{\it (Encyclopedia of Mathematics and Its Applications {\rm vol  71})} 
ed G C  Rota (Cambridge: Cambridge University Press)


  

\bibitem{sfagt} Beals R and Wong R 2010 {\it Special Functions: A
  Graduate Text} (Cambridge: Cambridge University Press)


\bibitem{kdf2} Exton, H 1978 {\it Handbook of hypergeometric
  integrals, Mathematics and its Applications} (Chichester: Ellis
  Horwood Ltd.)
  

  
\bibitem{cipra1998} 
 Cipra B A 1998 
{A new testament for special functions?} 
 {\it  SIAM News} {\bf 31(2)} 




%
%
%
%
%





\bibitem{hucht2011}
Hucht A, Grüneberg D and  Schmidt F M 2011
Aspect-ratio dependence of thermodynamic Casimir forces
{\it Phys. Rev. E} {\bf 83} 
051101 



\bibitem{jstat2015} Viswanathan G M, 2015
%
  { 
The hypergeometric series for the
partition function of the 2D Ising model} 
{\it J.  Stat.  Mech.} {\bf 2015} P04007




  
\bibitem{scirep} 
Siudem G, Fronczak A and Fronczak P 2016 
  %
Exact low-temperature series expansion for the partition function of
the zero-field Ising model on the infinite square lattice
%
{\it Sci. Rep.} {\bf  6}
33523


\bibitem{kdf1}
  Kampé de Fériet, J 1937  {La fonction hypergéométrique} {\it Mémorial
des sciences mathématiques} {\bf 85} (Paris: Gauthier-Villars) 




\bibitem{pc}
Glasser M L, personal e-mail communication,
  12 November 2014 

  
\bibitem{feynman} 
Feynman
R P, 1972
{\it Statistical Mechanics. A set of
  lectures} 
(Reading: Benjamin and Cummings Publishing)



\bibitem{montroll} 
Newell G F
and 
Montroll E W 1953
{On the Theory of the Ising Model of Ferromagnetism},
{\it   Rev. Mod. Phys.} {\bf 25}  353 


\bibitem{gr}
  Guttmann A J and Rogers M D 2012 
%
Spanning tree generating functions and Mahler measures
%
  J. Phys. A: Math. Theor. {\bf 45} 494001
  
\bibitem{pre2017} Viswanathan G M 2017
Correspondence between spanning trees and the Ising model on a square lattice
  {\it Phys. Rev. E} {\bf 95}
  062138 

\bibitem {domb}
 Domb  C 1949  Order-disorder statistics. ii. a two-dimensional model
  {\it Proceedings of the Royal Society of London. Series A. Mathematical
    and Physical Sciences} {\bf 199}  199
  %
%

  
\bibitem {enc}
OEIS Foundation Inc. 2019 The On-Line Encyclopedia of Integer Sequences, https://oeis.org/A002890
  
\end{thebibliography}
\end{document}